\documentclass[12pt]{article}

\usepackage{amsmath,amsfonts,amssymb,amsthm,graphicx,cite}
\usepackage{mathrsfs}

\numberwithin{equation}{section}

\begin{document}

\newtheorem{theorem}{Theorem}[section]
\newtheorem{corollary}[theorem]{Corollary}
\newtheorem{lemma}[theorem]{Lemma}
\newtheorem{proposition}[theorem]{Proposition}

\newcommand{\adiffop}{A$\Delta$O}
\newcommand{\adiffops}{A$\Delta$Os}

\newcommand{\be}{\begin{equation}}
\newcommand{\ee}{\end{equation}}
\newcommand{\bea}{\begin{eqnarray}}
\newcommand{\eea}{\end{eqnarray}}
\newcommand{\sh}{{\rm sh}}
\newcommand{\ch}{{\rm ch}}
\newcommand{\De}{\Delta}
\newcommand{\de}{\delta}
\newcommand{\te}{\tilde{e}}
\newcommand{\ty}{\tilde{y}}
\newcommand{\Z}{{\mathbb Z}}
\newcommand{\N}{{\mathbb N}}
\newcommand{\C}{{\mathbb C}}
\newcommand{\Cs}{{\mathbb C}^{*}}
\newcommand{\R}{{\mathbb R}}
\newcommand{\Q}{{\mathbb Q}}
\newcommand{\T}{{\mathbb T}}
\newcommand{\re}{{\rm Re}\, }
\newcommand{\im}{{\rm Im}\, }
\newcommand{\cW}{{\cal W}}
\newcommand{\cJ}{{\cal J}}
\newcommand{\cE}{{\cal E}}
\newcommand{\cA}{{\cal A}}
\newcommand{\cR}{{\cal R}}
\newcommand{\cP}{{\cal P}}
\newcommand{\cM}{{\cal M}}
\newcommand{\cN}{{\cal N}}
\newcommand{\cI}{{\cal I}}
\newcommand{\cMs}{{\cal M}^{*}}
\newcommand{\cB}{{\cal B}}
\newcommand{\cD}{{\cal D}}
\newcommand{\cC}{{\cal C}}
\newcommand{\cG}{{\cal G}}
\newcommand{\cL}{{\cal L}}
\newcommand{\cF}{{\cal F}}
\newcommand{\cH}{{\cal H}}
\newcommand{\cO}{{\cal O}}
\newcommand{\cS}{{\cal S}}
\newcommand{\cT}{{\cal T}}
\newcommand{\cU}{{\cal U}}
\newcommand{\cQ}{{\cal Q}}
\newcommand{\cV}{{\cal V}}
\newcommand{\cK}{{\cal K}}
\newcommand{\cZ}{{\cal Z}}

\newcommand{\rE}{{\rm E}}
\newcommand{\rF}{{\rm F}}
\newcommand{\rI}{{\rm I}}
\newcommand{\rR}{{\rm R}}

\newcommand{\fm}{\mathfrak{m}}
\newcommand{\intR}{\int_{-\infty}^{\infty}}
\newcommand{\intI}{\int_{0}^{\pi/2r}}
\newcommand{\limp}{\lim_{\re x \to \infty}}
\newcommand{\limn}{\lim_{\re x \to -\infty}}
\newcommand{\limpn}{\lim_{|\re x| \to \infty}}
\newcommand{\diag}{{\rm diag}}
\newcommand{\Ln}{{\rm Ln}}
\newcommand{\Arg}{{\rm Arg}}
\newcommand{\LHP}{{\rm LHP}}
\newcommand{\RHP}{{\rm RHP}}
\newcommand{\UHP}{{\rm UHP}}
\newcommand{\Res}{{\rm Res}}
\newcommand{\ep}{\epsilon}
\newcommand{\ga}{\gamma}
\newcommand{\sing}{{\rm sing}}

\title{Joint eigenfunctions for the relativistic Calogero-Moser Hamiltonians of hyperbolic type.\\ III.  Factorized asymptotics}
 
\author{Martin Halln\"as\footnote{E-mail: hallnas@chalmers.se} \\ Department of Mathematical Sciences \\ Chalmers University of Technology and the University of Gothenburg\\ SE-412 96 Gothenburg, Sweden \\ and \\Simon Ruijsenaars\footnote{E-mail: siruleeds@gmail.com} \\ School of Mathematics \\ University of Leeds \\ Leeds LS2 9JT, UK}

\date{\today}

\maketitle

\begin{abstract}
In the two preceding parts of this series of papers, we introduced and studied a recursion scheme for constructing joint eigenfunctions $J_N(a_+, a_-,b;x,y)$ of the Hamiltonians arising in the integrable $N$-particle systems of hyperbolic relativistic Calogero-Moser type. We focused on the first steps of the scheme in Part I, and on the cases $N=2$ and $N=3$ in Part II. In this paper, we determine the dominant asymptotics of a similarity transformed function $\rE_N(b;x,y)$ for $y_j-y_{j+1}\to\infty$, $j=1,\ldots, N-1$, and thereby confirm the long standing conjecture that the particles in the hyperbolic relativistic Calogero-Moser system exhibit soliton scattering. This result generalizes a main result in Part II to all particle numbers $N>3$.  
\end{abstract}

\tableofcontents

\section{Introduction}
In the first part \cite{HR14} of this series of papers, we presented and developed the first steps in a recursion scheme for constructing joint eigenfunctions $J_N(a_+,a_-,b;x,y)$ for the commuting analytic difference operators (henceforth A$\De$Os) associated with the integrable $N$-particle systems of hyperbolic relativistic Calogero-Moser type. More specifically, we presented the formal features of the scheme, explicitly demonstrated its arbitrary-$N$ viability for the `free' cases and established holomorphy domains and uniform decay bounds that were sufficient to render the scheme rigorous. Motivated by results on the `free' cases as well as the $N=2$ case, which can be gleaned from \cite{R11}, we also detailed several conjectured features of the joint eigenfunctions $J_N$.

In the second part \cite{HR18}, we proved a number of these conjectures in the cases $N=2$ and $N=3$. Indeed, we established global meromorphy, a number of invariance properties and a duality relation, and undertook a detailed study of asymptotic behavior. The purpose of this third part is to generalize the results on  asymptotics to all particle numbers $N>3$. We shall make use of previous results in this series of papers without further ado, referring back to sections and equations in \cite{HR14} and \cite{HR18} by using the prefix I and II, respectively.

To a large extent, we can follow our approach in the $N=3$ case, but the technical difficulties we encounter are considerably more involved. Important auxiliary results have been isolated in Lemma \ref{Lem:auxN} and Theorem \ref{Thm:IPLbou}. The latter theorem allows us to avoid the use of the bound II~(2.73) on $E_2$ that we used for the $N=3$ case, cf.~the proof of II~Theorem~3.7. This amounts to one of several simplifications of our $N=3$ results in II~Section 3. We could not obtain a counterpart of the bound II (2.73) for $E_N$ with $N>2$, but fortunately Theorem \ref{Thm:IPLbou} obviates this snag as well.

In order to describe the results and organization of this paper in more detail, we need to first recall the construction of $J_N$ from $J_{N-1}$. Following I and II, we take $a_+,a_-\in(0,\infty)$, use further parameters
\be\label{aconv}
\alpha\equiv 2\pi/a_+a_-,\ \ \ \ a\equiv (a_++a_-)/2,
\ee 
\be\label{asl}
a_s\equiv\min (a_+,a_-),\ \ \ a_l\equiv\max (a_+,a_-),
\ee
and work with $b$-values in the strip
\be
S_a\equiv \{b\in\C \mid \re b\in (0,2a)\}.
\ee
In addition, we make extensive use of the generalized Harish-Chandra $c$-function
\be\label{c}
c(b;z)\equiv \frac{G(z+ia-ib)}{G(z+ia)}=c(b;-z-2ia+ib),
\ee
and its multivariate version
\be\label{CN}
C_N(b;x)\equiv \prod_{1\leq j<k\leq N}c(b;x_j-x_k),\ \ \ N\geq 2.
\ee
Here $G(z)\equiv G(a_+,a_-;z)$ denotes the hyperbolic gamma function, whose salient features are reviewed in I Appendix A and II Appendix A. In particular, in \eqref{c} and frequently below, we use the reflection equation $G(-z)=1/G(z)$, cf.~I (A.6). (To unburden notation, we usually suppress the dependence on the parameters $a_+,a_-$; also, the dependence on $N$ and $b$ is often omitted when no ambiguities are likely to arise.)

In the  construction of $J_N$ from $J_{N-1}$ in I Section 6, we encountered the integrand
\be\label{IN}
I_N(b;x,y,z)\equiv  W_{N-1}(b;z)\cS^\sharp_N(b;x,z)J_{N-1}(b;z,(y_1-y_N,\ldots,y_{N-1}-y_N)),
\ee
with weight function
\be\label{WN}
W_N(b;z)\equiv 1/C_N(b;z)C_N(b;-z),
\ee
and kernel function (cf.~I (A.6))
\be\label{cSN}
\begin{split}
\cS^\sharp_N(b;x,z) &\equiv \prod_{j=1}^N\prod_{k=1}^{N-1}\frac{G(z_k-x_j-ib/2)}{G(z_k-x_j+ib/2)}\\
&= \prod_{j=1}^N\prod_{k=1}^{N-1}c(b;z_k-x_j-ia+ib/2).
\end{split}
\ee
More precisely, I (6.6) yields the representation
\be\label{JN}
J_N(b;x,y) =  \frac{\exp(i\alpha y_N(x_1+\cdots+x_N))}{(N-1)!}\int_{\R^{N-1}} dz\, I_N(b;x,y,z),\ \ \ b\in S_a,\ \ x,y\in\R^N.
\ee

Defining
\be\label{XY}
X_N\equiv \frac{1}{N}\sum_{j=1}^Nx_j,\ \ Y_N\equiv \frac{1}{N}\sum_{j=1}^Ny_j,\ \  x^{(N)}_j\equiv  x_j-X_N,\ \ \ y^{(N)}_j\equiv  y_j-Y_N,\ \ \ j=1,\ldots,N,
\ee
a straightforward induction argument revealed another important representation that we have occasion to invoke below, namely,
\begin{align}
\label{JNrep2}
J_N(x,y) &= \exp(Ni\alpha X_NY_N)J_N^r(x,y),\\
\label{JNr}
J_N^r(x,y) &\equiv \frac{1}{(N-1)!}\int_{\R^{N-1}}dz\, W_{N-1}(z)\cS^\sharp_N(x^{(N)},z)J_{N-1}(z,(y_1-y_N,\ldots,y_{N-1}-y_N)),
\end{align}
cf.~I (6.27)--(6.28). Note that the function $J_N^r(x,y)$ depends only on the differences $x_j-x_{j+1}$ and $y_j-y_{j+1}$, $j=1,\ldots,N-1$.

By performing simultaneous contour shifts in the former representation \eqref{JN}, we showed in I Theorem 6.1 that for fixed $y\in\R^N$ the function $J_N(b;x,y)$ is holomorphic in
\be\label{DN}
D_N\equiv\Big\{(b,x)\in S_a\times\C^N\mid \max_{1\leq j<k\leq N}|\im(x_j-x_k)|<2a-\re b\Big\}.
\ee
Moreover, after restricting attention to a subdomain of $D_N$ for the dependence on $(b,x)$, we could allow $y\in\C^N$ such that $|\im(y_j-y_k)|<\re b$, $1\leq j<k\leq N$. Specifically, introducing the restricted domain
\be\label{DNr}
D_N^r\equiv \{ (b,x)\in S_a\times \C^N \mid |\im x^{(N)}_j|<a-\re b/2,\ \ j=1,\ldots,N \}\subset D_N,
\ee
we used the latter representation \eqref{JNrep2} to prove that $J_N(b;x,y)$ is holomorphic in $(b,x,y)$ on the domain
\be\label{cDN}
\cD_N\equiv \Big\{ (b,x,y)\in D_N^r\times \C^N \mid  \max_{1\leq j<k\leq N}|\im (y_j-y_k)|<\re b \Big\},
\ee
cf.~I Theorem 6.4.

In Section \ref{Sec2} we study the asymptotic behavior of the function
\be\label{rEN}
\rE_N(b;x,y)\equiv \left(\frac{\phi(b)G(ib-ia)}{\sqrt{a_+a_-}}\right)^{N(N-1)/2}\frac{J_N(b;x,y)}{C_N(b;x)C_N(2a-b;y)},
\ee
where
\be\label{phi}
\phi(b)\equiv \exp(i\alpha b(b-2a)/4)=\phi(2a-b).
\ee
Since the $c$-function is not even, $\rE_N$ lacks some of the invariance properties of $J_N$. However, the multipliers in~\eqref{rEN} are meromorphic functions whose features are known in great detail. Hence the analyticity properties of~$E_N$  follow from those of $J_N$. Moreover, $E_N$ is particularly well suited for Hilbert space purposes. 

As the principal result of Section \ref{Sec2} and of this paper, we prove in Theorem \ref{Thm:rENas} that $\rE_N$ has the `unitary asymptotics' 
\be\label{ENsc}
\rE_N(b;x,y)\sim \rE_N^{{\rm as}}(b;x,y)\equiv \sum_{\sigma\in S_N}\prod_{\substack{j<k\\\sigma^{-1}(j)>\sigma^{-1}(k)}}(-u(b;x_k-x_j))\cdot\exp\Big(i\alpha \sum_{j=1}^N x_{\sigma(j)}y_j\Big),
\ee
for $y_j-y_{j+1}\to\infty$, $j=1,\ldots, N-1$. Here the scattering function $u$ is given by
\be\label{u}
u(b;z)\equiv -\frac{c(b;z)}{c(b;-z)}=-\prod_{\de=+,-}\frac{G(z+\de i(a-b))}{G(z+\de ia)}.
\ee
It clearly satisfies
\be\label{uref}
u(b;z)u(b;-z)=1,
\ee
and we also have
\be\label{uphase}
|u(b;z)|=1,\ \ \ \ b,z\in\R,
\ee
 due to the reflection equation I (A.6) and the conjugation relation I (A.9). Moreover, we obtain a uniform bound on $\rE_N(x,y)$ for suitably restricted $(x,y)\in\C^N\times\R^N$, which plays a crucial role in the inductive step $N-1\to N$.
 
The asymptotic behavior \eqref{ENsc} confirms a long-standing conjecture. In physical parlance, it says that the particles in the relativistic Calogero-Moser systems of hyperbolic type exhibit soliton scattering (conservation of momenta and factorization of the $S$-matrix), cf.~I Section 7. For a survey of the $A_{N-1}$ type Calogero-Moser systems and their relation to soliton PDEs we refer to~\cite{R94}. In particular, the sine-Gordon soliton scattering corresponds to choosing $b$ equal to $a_{+}/2$ or $a_{-}/2$ in~\eqref{u}. See also the recent paper~\cite{HwR16} for more information on this `sine-Gordon' perspective.
 
Within the context of harmonic analysis, factorized asymptotics  was first established by Harish-Chandra for the spherical functions associated with certain symmetric spaces.  Viewed from the $A_{N-1}$ perspective of this paper, the Harish-Chandra work pertains to the nonrelativistic Calogero-Moser systems for a few special coupling constants (see~\cite{H00} for a comprehensive account of the general Harish-Chandra results, as well as related ones, and \cite{OP83} for their relevance to Calogero-Moser systems). Factorized asymptotics for the hyperbolic case with arbitrary positive coupling was first proved by Opdam~\cite{O95}, working within the arbitrary root system context developed by him and Heckman, a summary of which can be found in~\cite{HS94}. A crucial aspect of the asymptotic analysis in these references is the existence and exploitation of series expansions. By contrast, no such expansions are known for the eigenfunctions at issue in this paper. As in our previous work, a key point is rather to use their recursive structure.

\section{Asymptotic behavior}\label{Sec2}

Using II Theorems 3.7--3.8 as the starting point for an induction argument, we proceed to determine the asymptotics of the function $\rE_N(b;x,y)$ \eqref{rEN} for $d_N(y)\to\infty$, where
\be
d_N(y)\equiv\min_{1\leq j<k\leq N}(y_j-y_k),\ \ \ y\in\R^N.
\ee
More specifically, Theorems \ref{Thm:rENas}--\ref{Thm:ubound} below are a consequence of the former for $N = 3$, and our induction assumption is that they hold true if we replace $N$ by $N-1$. In the present general-$N$ setting, however, we restrict attention to $\re b$ varying over a sub-interval of $(0,2a)$, namely $(0,a_l]$. Thus we introduce the strip
\be\label{Sl}
S_l\equiv \{ b\in\C\mid \re b\in (0,a_l]\}.
\ee

We start with some auxiliary results about $J_N(b;x,y)$.

\begin{proposition}\label{Prop:JNsym}
For fixed $y\in\R^N$, the function $J_N(b;x,y)$ is holomorphic in
\be\label{DNl}
D_N^l\equiv\Big\{(b,x)\in S_l\times\C^N\mid \max_{1\leq j<k\leq N}|\im(x_j-x_k)|<a_s\Big\}.
\ee
Furthermore, for all $(b,x,y)\in\cD_N$~\eqref{cDN} and $\eta\in\C$, we have symmetry properties
\be\label{JNri}
J_N(x,y)=J_N(-x,-y),
\ee 
\be\label{JNhom}
\begin{split}
J_N(x,y)=& \exp(-i\alpha\eta(y_1+\cdots+y_N))J_N((x_1+\eta,\ldots,x_N+\eta),y)\\
& =\exp(-i\alpha\eta(x_1+\cdots+x_N))J_N(x,(y_1+\eta,\ldots,y_N+\eta)),
\end{split}
\ee
\be\label{JNp}
J_N(\sigma x,y)=J_N(x,y),\ \ \ \sigma\in S_N.
\ee
\end{proposition}

\begin{proof}
The first assertion is an easy consequence of the readily verified inclusion
\be
D_N^l\subset D_N,
\ee
cf.~\eqref{DN}. 

Letting $x,y\in\R^N$ to begin with, the permutation invariance \eqref{JNp} is immediate from the defining representation \eqref{JN}. To establish the invariance properties \eqref{JNri}--\eqref{JNhom}, we assume inductively that they hold true for $N\geq 3$. (In the case $N=3$ this is the content of II Proposition 3.1.)  From \eqref{WN}--\eqref{cSN}, \eqref{JNri} with $N\to N-1$ and the reflection equation I (A.6) for $G(z)$, we infer
\be\label{INri}
I_N(-x,-y,-z)=I_N(x,y,z).
\ee
Changing variable $z\to -z$ in the representation \eqref{JN}, the invariance property \eqref{JNri} is a direct consequence of \eqref{INri}. Requiring in addition $\eta\in\R$, we deduce \eqref{JNhom} from the alternative representation given by \eqref{JNrep2}--\eqref{JNr}. Since \eqref{JNri}--\eqref{JNp} are preserved under analytic continuation, the proof is complete.
\end{proof}

This proposition has the following corollary.
 
\begin{corollary}\label{Cor:EN}
Letting $y\in\R^N$, the function $\rE_N(b;x,y)$ is meromorphic in $D_N^l$ and holomorphic in
\be
D_{N,\beta}^l\equiv\Big\{(b,x)\in D_N^l \mid \im(x_j-x_{j+1})<\beta,\ \ j=1,\ldots,N-1,\ \ \im(x_1-x_N)>-a_s\Big\},
\ee
where
\be\label{beta}
\beta\equiv \min (\re b,a_s).
\ee
Moreover, for all $(b,x,y)\in\cD_N$~\eqref{cDN} and $\eta\in\C$, it satisfies
\be\label{rENri}
\rE_N(-x,-y)=\rE_N(x,y)\prod_{1\le j<k\le N}u(x_j-x_k)u(y_j-y_k), 
\ee
\be\label{rENhom}
\begin{split}
\rE_N(x,y)=& \exp(-i\alpha\eta(y_1+\cdots+y_N))\rE_N((x_1+\eta,\ldots,x_N+\eta),y)\\
& =\exp(-i\alpha\eta(x_1+\cdots+x_N))\rE_N(x,(y_1+\eta,\ldots,y_N+\eta)),
\end{split}
 \ee
\be\label{rENp}
\rE_N(\sigma x,y)=\rE_N(x,y)\prod_{\substack{j<k\\ \sigma^{-1}(j)>\sigma^{-1}(k)}}(-u(x_j-x_k)),\ \ \ \sigma\in S_N,
\ee
where $(\sigma x)_j\equiv x_{\sigma(j)}$. 
\end{corollary}

\begin{proof}
The zeros of $C_N(b;x)$ are located at
\be\label{CNzero}
x_j-x_k=-2ia-ima_+-ina_-,\ ib+ima_++ina_-, \ 1\leq j<k\leq N-1,\ \ m,n\in\N,
\ee
so the poles of $1/C_N(b;x)$  do not belong to~$D_{N,\beta}^l$. Hence the first assertion is clear from the relation \eqref{rEN} between $J_N$ and $\rE_N$.

Keeping in mind \eqref{CN} and \eqref{u}, the symmetry features are readily inferred from \eqref{rEN} and Proposition \ref{Prop:JNsym}. 
\end{proof}

Recalling from I (2.11) the kernel function
\be\label{cKN}
\cK^\sharp_N(b;x,z)\equiv [C_N(b;x)C_{N-1}(b;-z)]^{-1}\cS^\sharp_N(b;x,z),
\ee
it is easily seen that \eqref{IN}--\eqref{JN} and \eqref{rEN} yield the representation
\begin{multline}\label{rENrep}
\rE_N(b;x,y)=\frac{1}{(N-1)!}\left(\frac{\phi(b)G(ib-ia)}{\sqrt{a_+a_-}}\right)^{N-1}\\
\times\frac{\exp(i\alpha y_N(x_1+\cdots+x_N))}{\prod_{n=1}^{N-1}c(2a-b;y_n-y_N)}\int_{\R^{N-1}}dz\, \rI_N(b;x,y,z),\ \ \ b\in S_a,\ \ x,y\in\R^N,
\end{multline}
with integrand
\be\label{rIN}
\rI_N(b;x,y,z)\equiv \cK_N^\sharp(b;x,z)\rE_{N-1}(b;z,(y_1-y_N,\ldots,y_{N-1}-y_N)).
\ee

Following our treatment of the $N=2$ and $N=3$ cases in II, we determine the dominant asymptotics of $\rE_N$ by shifting the $z_k$-contours $\R$ in \eqref{rENrep} up past the poles of $\rI_N$ located at
\be\label{psN}
z_k=x_j+ia-ib/2,\ \ \ k=1,\ldots,N-1,\ \ j=1,\ldots,N.
\ee
Using \eqref{rEN} and \eqref{c}--\eqref{CN}, we find that the $G$-zero $G(ia)=0$ (cf.~I (A.12)) ensures that $\rE_N$ vanishes whenever $x_j=x_k$, $1\leq j<k\leq N$. Hence no generality is lost by assuming
\be\label{diffxN}
x_j\neq x_k,\ \ \ 1\leq j<k\leq N,
\ee
so that the poles \eqref{psN} are simple.

In order to keep track of the residues that appear, it will be important to shift the $N-1$ contours one at a time. Doing so, we must ensure that we retain sufficient decay of $\rI_N$ on the contour tails and that we do not meet any of its $x$-independent poles. 

To control the tail decay, we first use the $c$-definition~\eqref{c} and the $G$-asymptotics specified in I~(A.14)--(A.16) to infer
\be\label{casym}
|\phi(b)^{\mp 1}\exp(\pm\alpha bz/2)c(b;z)-1| \le C_1(\rho,b,\im z)\exp(-\alpha\rho|\re z|),\ \ \ \re z\to\pm\infty,
\ee
where the decay rate $\rho$ can be chosen in $[a_s/2,a_s)$, and where $C_1$ is continuous on $[a_s/2,a_s)\times S_a\times \R$. 

Next, by the induction assumption, we may invoke Theorem~\ref{Thm:ubound} with $N\to N-1$. Requiring at first $\im(z_j-z_k)\in(-a_s,0]$, $1\leq j<k\leq N-1$, we can use the resulting bound on $\rE_{N-1}$, together with \eqref{cSN} and~\eqref{casym}, to deduce that the integrand $\rI_N$ decays exponentially for $|\re z_k|\to\infty$. Indeed, we have $N-1$ factors of the form $c(z_k\cdots)$ in the numerator and $N-2$ factors of the form $c(z_k\cdots)$ or $c(-z_k\cdots)$ in the denominator, cf.~\eqref{cSN} and~\eqref{CN} with $N\to N-1$.

 Now from~\eqref{casym} and the $u$-definition~\eqref{u} we readily obtain
\be\label{uasym}
|u(b;z)\phi(b)^{\mp 2}+1|
\le C_2(\rho,b,\im z)\exp(-\alpha\rho|\re z|),\ \ \ \re z\to\pm\infty,
\ee
 with $C_2$  continuous on $[a_s/2,a_s)\times S_a\times \R$.
Furthermore, using \eqref{rIN}, \eqref{cKN} and \eqref{rENp}, we find
\be\label{rINp}
\rI_N(x,y,\tau z)=\rI_N(x,y,z),\ \ \ \tau\in S_{N-1}.
\ee
Combining this with~\eqref{uasym}, we conclude that $\rI_N$ has the same decay for $\im(z_j-z_k)\in[0,a_s)$, $1\leq j<k\leq N-1$. 

The upshot of this analysis is that the shift of a single contour causes no problems at the tail ends, as long as the contours are separated by a distance less than $a_s$. Moreover, since we require $b\in S_l$, the $x$-independent poles of $\rI_N$ are not met for $|\im(z_j-z_k)|<\beta$, $1\leq j<k\leq N-1$, cf.~Corollary~\ref{Cor:EN}.  
 
Finally, for a given vector $t\equiv(t_1,\ldots,t_M)\in\C^M$, $M>1$, we use the notation 
\be
t(\nu_1,\ldots,\nu_L),\ \ \ 1\leq \nu_j\neq\nu_k\leq M,\ \ 1\leq j<k\leq L,
\ee
to denote the vector in $\C^{M-L}$ obtained by omitting the entries $t_{\nu_1},\ldots,t_{\nu_L}$ in $t$. Introducing the additional notation
\be
z_{>L}\equiv z(z_1,\ldots,z_L)=(z_{L+1},\ldots,z_{N-1}),\ \ \ L=1,\ldots,N-2,
\ee
and the functions
\be\label{MN}
M_N(b;y)\equiv \frac{\phi(b)^{N-1}}{\prod_{n=1}^{N-1}c(2a-b;y_n-y_N)}\rho_N(b;y),
\ee
\be\label{rhoN}
\rho_N(b;y)\equiv \exp\Big(-\alpha(a-b/2)\sum_{n=1}^{N-1}(y_n-y_N)\Big),
\ee
we are now ready to implement the contour shift procedure.

\begin{lemma}\label{Lem:auxN}
Letting $(r,b)\in(0,a_s)\times S_l$ and $x,y\in\R^N$ with the $x$-restriction \eqref{diffxN} in effect, we have
\begin{multline}\label{rENrep2}
\frac{\rE_N(x,y)}{M_N(y)}\exp(-i\alpha y_N(x_1+\cdots+x_N))\\
=\frac{1}{\rho_N(y)}\Bigg[\frac{1}{(N-1)!}\left(\frac{G(ib-ia)}{\sqrt{a_+a_-}}\right)^{N-1}\int_{(C_b+ir)^{N-1}}dz\, \rI_N(x,y,z)\\
+\sum_{L=1}^{N-2}\frac{1}{(N-1-L)!}\left(\frac{G(ib-ia)}{\sqrt{a_+a_-}}\right)^{N-1-L}\sum_{1\leq\nu_1<\cdots<\nu_L\leq N}\cU_{\nu_1,\ldots,\nu_L}(x)\\
\times \int_{(C_b+ir)^{N-1-L}}dz_{>L}\, \hat{\rI}_{N;\nu_1,\ldots,\nu_L}(x,y,z_{>L})\Bigg]\\
+\sum_{\nu=1}^N\frac{C_N(x(\nu),x_\nu)}{C_N(x)}\rE_{N-1}(x(\nu),(y_1-y_N,\ldots,y_{N-1}-y_N)).
\end{multline}
Here, $\rI_N(x,y,z)$ is given by \eqref{rIN}, we have set
\begin{multline}
\hat{\rI}_{N;\nu_1,\ldots,\nu_L}(b;x,y,z_{>L})\equiv \cK^\sharp_{N-L}(b;x(\nu_1,\ldots,\nu_L),z_{>L})\\
\times \rE_{N-1}(b;(x_{\nu_1}+ia-ib/2,\ldots,x_{\nu_L}+ia-ib/2,z_{>L}),(y_1-y_N,\ldots,y_{N-1}-y_N)),
\end{multline}
\be
\cU_{\nu_1,\ldots,\nu_L}(b;x)\equiv \prod_{\ell=1}^L\prod_{\substack{j<\nu_\ell\\ j\neq\nu_1,\ldots,\nu_{\ell-1}}}(-u(b;x_{\nu_\ell}-x_j)),
\ee
and $C_b$ denotes  the contour
\be
C_b\equiv \R+i(a-\re b/2).
\ee
\end{lemma}

\begin{proof}
To start with, we write the left-hand side of \eqref{rENrep2} as
\be
\frac{1}{(N-1)!}\frac{1}{\rho_N(y)}\cG^{N-1}\int_{\R^{N-1}}dz\, \cK^\sharp_N(x,z)\rE_{N-1}(z,\hat{y}),
\ee
cf.~\eqref{rENrep}--\eqref{rIN} and \eqref{MN}. Here we have introduced 
\be
\hat{y}\equiv (y_1-y_N,\ldots,y_{N-1}-y_N),\ \ \ \cG\equiv \frac{G(ib-ia)}{\sqrt{a_+a_-}}.
\ee
We find it convenient to work at first with $J_{N-1}(z,\hat{y})$, since it is $S_{N-1}$-invariant in $z$. Therefore, we use \eqref{rEN} with $N\to N-1$ to get (cf.~\eqref{cKN} and \eqref{WN})
\be\label{LHSExp}
\frac{1}{(N-1)!}\frac{1}{\rho_N(b;y)}\cG^{N-1}\big(\phi(b)\cG\big)^{(N-1)(N-2)/2}\frac{1}{C_N(b;x)}\frac{\cL_N(x,y)}{C_{N-1}(2a-b;\hat{y})},
\ee
with
\be
\cL_N(b;x,y)\equiv \int_{\R^{N-1}}dz\, W_{N-1}(b;z)\cS^\sharp_N(b;x,z)J_{N-1}(b;z,\hat{y}).
\ee

Letting
\be\label{ep}
0<\epsilon<\beta/2,
\ee
(with $\beta$ defined by~\eqref{beta}), we move the $N-1$ contours $\R$ simultaneously to $C_b-i\epsilon$ without meeting poles. Shifting the $z_1$-contour to $C_b+i\epsilon$, we pick up residues at the poles \eqref{psN} with $k=1$. These poles arise from the factor
\be
c(z_1-x_j-ia+ib/2)=G(z_1-x_j-ib/2)G(x_j-z_1-ib/2)
\ee
in $\cS^\sharp_N(x,z)$ \eqref{cSN}, and the assumption \eqref{diffxN} ensures that they are simple. Recalling the $G$-residue I (A.13), we have
\be
\lim_{z_1\to x_j+ia-ib/2}(z_1-x_j-ia+ib/2)G(x_j-z_1-ib/2)=\lim_{z\to -ia}(-z-ia)G(z)=\frac{\sqrt{a_+a_-}}{2\pi i},
\ee
so that
\be
2\pi i\, \Res\, c(z_1-x_j-ia+ib/2)\arrowvert_{z_1=x_j+ia-ib/2}=\frac{\sqrt{a_+a_-}}{G(ib-ia)}=\cG^{-1}.
\ee
Thus we infer that $\cL_N$ is given by
\begin{multline}\label{cLNExpr2}
\cL_N(x,y)=\int_{C_b+i\epsilon}dz_1\, \int_{(C_b-i\epsilon)^{N-2}}dz_{>1}\, W_{N-1}(z)\cS^\sharp_N(x,z)J_{N-1}(z,\hat{y})\\
+\cG^{-1}\int_{(C_b-i\epsilon)^{N-2}}dz_{>1} \sum_{\nu_1=1}^N \cR_{\nu_1}(x,z_{>1})J_{N-1}((x_{\nu_1}+ia-ib/2,z_{>1}),\hat{y}),
\end{multline}
with remainder residue
\bea\label{cR1}
\cR_{\nu_1}(x,z_{>1}) &= & \prod_{\substack{m,n=2\\m\neq n}}^{N-1}\frac{1}{c(z_m-z_n)}\cdot \prod_{n=2}^{N-1}\frac{1}{c(x_{\nu_1}-z_n+ia-ib/2)c(z_n-x_{\nu_1}-ia+ib/2)}
\nonumber \\
&\times & \prod_{j=1}^N\prod_{k=2}^{N-1}c(z_k-x_j-ia+ib/2)\cdot \prod_{\substack{j=1\\ j\neq\nu_1}}^N c(x_{\nu_1}-x_j)
\nonumber  \\
&=& W_{N-2}(z_{>1})\prod_{\substack{j=1\\ j\neq\nu_1}}^N\prod_{k=2}^{N-1}c(z_k-x_j-ia+ib/2)
\nonumber  \\
& \times &\frac{\prod_{\substack{j=1\\ j\neq\nu_1}}^N c(x_{\nu_1}-x_j)}{\prod_{k=2}^{N-1}c(x_{\nu_1}-z_k+ia-ib/2)}
\nonumber  \\
&=  & W_{N-2}(z_{>1})\cS^\sharp_{N-1}(x(\nu_1),z_{>1})\frac{\prod_{\substack{j=1\\ j\neq\nu_1}}^N c(x_{\nu_1}-x_j)}{\prod_{k=2}^{N-1}c(x_{\nu_1}-z_k+ia-ib/2)}.
\eea

We note that the $\epsilon$-choice \eqref{ep} guarantees that the factors $1/c(x_{\nu_1}-z_k+ia-ib/2)$ are analytic in $z_k$ for $|\im z_k-(a-\re b/2)|\leq\epsilon$. Hence, moving the $z_2$-contours in \eqref{cLNExpr2} up by $2\epsilon$, we only encounter the poles \eqref{psN} with $k=2$. In the residues spawned by the first integral we replace $z_1$ by $z_2$ and use the $S_{N-1}$-invariance of $J_{N-1}(z,\hat{y})$ in $z$ to obtain
\begin{multline}\label{cLNpart}
\int_{(C_b+i\epsilon)^2}dz_1dz_2\int_{(C_b-i\epsilon)^{N-3}}dz_{>2}\, W_{N-1}(z)\cS^\sharp_N(x,z)J_{N-1}(z,\hat{y})\\
+\cG^{-1}\int_{C_b+i\epsilon}dz_2\int_{(C_b-i\epsilon)^{N-3}}dz_{>2}\,\sum_{\nu_1=1}^N \cR_{\nu_1}(x,z_{>1})J_{N-1}((x_{\nu_1}+ia-ib/2,z_{>1}),\hat{y}).
\end{multline}
From the second integral in \eqref{cLNExpr2}, we get a copy of the second integral in \eqref{cLNpart} plus a residue term
\be
\cG^{-2}\int_{(C_b-i\epsilon)^{N-3}}dz_{>2} \sum_{\substack{\nu_1,\nu_2=1\\ \nu_1\neq\nu_2}}^N\cR_{\nu_1,\nu_2}(x,z_{>2})J_{N-1}((x_{\nu_1}+ia-ib/2,x_{\nu_2}+ia-ib/2,z_{>2}),\hat{y}),
\ee
which is readily determined by adapting the computations in \eqref{cR1}:
\be\label{cR2}
\cR_{\nu_1,\nu_2}(x,z_{>2})=W_{N-3}(z_{>2})\cS^\sharp_{N-2}(x(\nu_1,\nu_2),z_{>2})\prod_{\ell=1}^2\frac{\prod_{\substack{j=1\\ j\neq\nu_1,\nu_2}}^Nc(x_{\nu_\ell}-x_j)}{\prod_{k=3}^{N-1}c(x_{\nu_\ell}-z_k+ia-ib/2)}.
\ee
The upshot is that $\cL_N(x,y)$ can be written
\begin{multline}\label{cLNExpr3}
\cL_N(x,y)=\int_{(C_b+i\epsilon)^2}dz_1dz_2\int_{(C_b-i\epsilon)^{N-3}}dz_{>2}\, W_{N-1}(z)\cS^\sharp_N(x,z)J_{N-1}(z,\hat{y})\\
+2\cG^{-1}\int_{C_b+i\epsilon}dz_2\int_{(C_b-i\epsilon)^{N-3}}dz_{>2}\,\sum_{\nu_1=1}^N \cR_{\nu_1}(x,z_{>1})J_{N-1}((x_{\nu_1}+ia-ib/2,z_{>1}),\hat{y})\\
+\cG^{-2}\int_{(C_b-i\epsilon)^{N-3}}dz_{>2} \sum_{\substack{\nu_1,\nu_2=1\\ \nu_1\neq\nu_2}}^N\cR_{\nu_1,\nu_2}(x,z_{>2})J_{N-1}((x_{\nu_1}+ia-ib/2,x_{\nu_2}+ia-ib/2,z_{>2}),\hat{y}),
\end{multline}
with $\cR_{\nu_1}$ and $\cR_{\nu_1,\nu_2}$ given by \eqref{cR1} and \eqref{cR2}, respectively.

More generally, introducing the integration domains
\be
V_L^M\equiv (C_b+i\epsilon)^{M-L}\times (C_b-i\epsilon)^{N-1-M},\ \ \ 1\leq M\leq N-1,\ \ 0\leq L\leq M,
\ee
we claim that $\cL_N(x,y)$ can be written
\begin{multline}\label{cLNExpr4}
\cL_N(x,y)=\int_{V_0^M}dz\, W_{N-1}(z)\cS^\sharp_N(x,z)J_{N-1}(z,\hat{y})\\
+\sum_{L=1}^M\cG^{-L}\binom{M}{L}\int_{V_L^M}dz_{>L}\sum_{\substack{\nu_1,\ldots,\nu_L=1\\ \nu_j\neq\nu_k}}^N \cR_{\nu_1,\ldots,\nu_L}(x,z_{>L})\\
\times J_{N-1}((x_{\nu_1}+ia-ib/2,\ldots,x_{\nu_L}+ia-ib/2,z_{>L}),\hat{y}),
\end{multline}
for any $M=1,\ldots,N-1$. Here we have introduced
\begin{multline}\label{cRL}
\cR_{\nu_1,\ldots,\nu_L}(x,z_{>L}) \equiv W_{N-1-L}(z_{>L})\cS^\sharp_{N-L}(x(\nu_1,\ldots,\nu_L),z_{>L})\\
\times \prod_{\ell=1}^L\frac{\prod_{\substack{j=1\\ j\neq\nu_1,\ldots,\nu_L}}^Nc(x_{\nu_\ell}-x_j)}{\prod_{k=L+1}^{N-1}c(x_{\nu_\ell}-z_k+ia-ib/2)},\ \ \ L=1,\ldots,N-2, \ \ \  L\le M,
\end{multline}
whereas for $L=M=N-1$ the integral should be omitted and we have
\be
\cR_{\nu_1,\ldots,\nu_{N-1}}(x)\equiv \prod_{\ell=1}^{N-1} c(x_{\nu_\ell}-x_{\nu_N}),\ \ \ \{\nu_1,\ldots,\nu_N\}=\{1,\ldots,N\}.
\ee
By \eqref{cLNExpr2}--\eqref{cR1} and \eqref{cR2}--\eqref{cLNExpr3}, we know already that the claim holds true for $M=1,2$. Assuming \eqref{cLNExpr4} for $1\leq M\leq N-2$, we now prove its validity for $M\to M+1$.

To this end, we move the $z_{M+1}$-contours up by $2\epsilon$, meeting the simple poles
\be
z_{M+1}=x_{\nu_1}+ia-ib/2,\ \ \ \nu_1=1,\ldots,N,
\ee
in the first integral, and the simple poles
\be
z_{M+1}=x_{\nu_{L+1}}+ia-ib/2,\ \ \ \nu_{L+1}=1,\ldots,N,\ \ \nu_{L+1}\neq\nu_1,\ldots,\nu_L,
\ee
in the remaining integrals. Using $S_{N-1}$-invariance of $J_{N-1}(z,\hat{y})$ in~$z$, it is readily seen that the first integral yields, upon taking $z(M+1)\to z_{>1}$ in the residue integral,
\begin{multline}\label{intTerms1}
\int_{V_0^{M+1}}dz\, W_{N-1}(z)\cS^\sharp_N(x,z)J_{N-1}(z,\hat{y})\\
+\cG^{-1}\int_{V_1^{M+1}}dz_{>1}\sum_{\nu_1=1}^N\cR_{\nu_1}(x,z_{>1})J_{N-1}((x_{\nu_1}+ia-ib/2,z_{>1}),\hat{y}).
\end{multline}
Similarly, the $L$-summand with $L=1,\ldots,M$ yields, after taking $z_{>L}(M+1)\to z_{>L+1}$ in the residue integral,
\begin{multline}\label{intTerms2}
\cG^{-L}\binom{M}{L}\int_{V_L^{M+1}}dz_{>L}\sum_{\substack{\nu_1,\ldots,\nu_L=1\\ \nu_j\neq\nu_k}}^N \cR_{\nu_1,\ldots,\nu_L}(x,z_{>L})\\
\times J_{N-1}((x_{\nu_1}+ia-ib/2,\ldots,x_{\nu_L}+ia-ib/2,z_{>L}),\hat{y})\\
+\cG^{-L-1}\binom{M}{L}\int_{V_{L+1}^{M+1}}dz_{>L+1}\sum_{\substack{\nu_1,\ldots,\nu_{L+1}=1\\ \nu_j\neq\nu_k}}^N \cR_{\nu_1,\ldots,\nu_{L+1}}(x,z_{>L+1})\\
\times J_{N-1}((x_{\nu_1}+ia-ib/2,\ldots,x_{\nu_{L+1}}+ia-ib/2,z_{>{L+1}}),\hat{y}).
\end{multline}
Summing the terms \eqref{intTerms2} over $L=1,\ldots,M$ and adding the resulting expression to \eqref{intTerms1}, we arrive at the right-hand side of \eqref{cLNExpr4} with $M\to M+1$ by invoking Pascal's rule
\be
\binom{M}{L}+\binom{M}{L-1}=\binom{M+1}{L}.
\ee
Hence our claim is proved.

Next, we specialize \eqref{cLNExpr4} to $M=N-1$ and shift all contours up to $C_b+ir$ without encountering further poles. Using symmetry under permutations of $x_{\nu_1},\ldots,x_{\nu_L}$, we thus obtain
\begin{multline}\label{cLNExpr5}
\cL_N(x,y)=\int_{(C_b+ir)^{N-1}}dz\, W_{N-1}(z)\cS^\sharp_N(x,z)J_{N-1}(z,\hat{y})\\
+(N-1)!\sum_{L=1}^{N-2}\cG^{-L}\frac{1}{(N-1-L)!}\int_{(C_b+ir)^{N-1-L}}dz_{>L}\sum_{1\leq\nu_1<\cdots<\nu_L\leq N}\cR_{\nu_1,\ldots,\nu_L}(x,z_{>L})\\
\times J_{N-1}((x_{\nu_1}+ia-ib/2,\ldots,x_{\nu_L}+ia-ib/2,z_{>L}),\hat{y})\\
+(N-1)!\cG^{1-N}\sum_{1\leq\nu_1<\cdots<\nu_{N-1}\leq N}R_{\nu_1,\ldots,\nu_{N-1}}(x)J_{N-1}((x_{\nu_1}+ia-ib/2,\ldots,x_{\nu_{N-1}}+ia-ib/2),\hat{y}).
\end{multline}

In order to establish the representation \eqref{rENrep2}, we now  reformulate \eqref{cLNExpr5} in terms of $\rE_{N-1}$. From \eqref{rEN} and \eqref{CN}, we infer
\begin{multline}
J_{N-1}((x_{\nu_1}+ia-ib/2,\ldots,x_{\nu_L}+ia-ib/2,z_{>L}),\hat{y})\\
=(\phi(b)\cG)^{-(N-1)(N-2)/2}\rE_{N-1}((x_{\nu_1}+ia-ib/2,\ldots,x_{\nu_L}+ia-ib/2,z_{>L}),\hat{y})\\
\times C_{N-1}(2a-b;\hat{y})C_L(x_{\nu_1},\ldots,x_{\nu_L})C_{N-1-L}(z_{>L})\\
\times \prod_{\ell=1}^L\prod_{k=L+1}^{N-1}c(x_{\nu_\ell}-z_k+ia-ib/2).
\end{multline}
Combining \eqref{cRL} with \eqref{WN} and \eqref{cKN}, we deduce
\be
\begin{split}
\cR_{\nu_1,\ldots,\nu_L}(x,z_{>L}) &= \cK^\sharp_{N-L}(x(\nu_1,\ldots,\nu_L),z_{>L})\frac{C_{N-L}(x(\nu_1,\ldots,\nu_L))}{C_{N-1-L}(z_{>L})}\\
&\quad \times \prod_{\ell=1}^L\frac{\prod_{\substack{j=1\\ j\neq\nu_1,\ldots,\nu_L}}^Nc(x_{\nu_\ell}-x_j)}{\prod_{k=L+1}^{N-1}c(x_{\nu_\ell}-z_k+ia-ib/2)}.
\end{split}
\ee
It follows that
\begin{multline}\label{ratExpr}
\cR_{\nu_1,\ldots,\nu_L}(x,z_{>L})J_{N-1}((x_{\nu_1}+ia-ib/2,\ldots,x_{\nu_L}+ia-ib/2,z_{>L}),\hat{y})/C_{N-1}(2a-b;\hat{y})\\
=(\phi(b)\cG)^{-(N-1)(N-2)/2}\rE_{N-1}((x_{\nu_1}+ia-ib/2,\ldots,x_{\nu_L}+ia-ib/2,z_{>L}),\hat{y})\\
\times \cK^\sharp_{N-L}(x(\nu_1,\ldots,\nu_L),z_{>L})C_L(x_{\nu_1},\ldots,x_{\nu_L})C_{N-L}(x(\nu_1,\ldots,\nu_L))\\
\times \prod_{\ell=1}^L\prod_{\substack{j=1\\ j\neq\nu_1,\ldots,\nu_L}}^Nc(x_{\nu_\ell}-x_j).
\end{multline}
Since $\nu_1<\cdots<\nu_L$ in \eqref{cLNExpr5}, we can write
\be
\begin{split}\label{CNfact}
C_N(x) &= C_L(x_{\nu_1},\ldots,x_{\nu_L})C_{N-L}(x(\nu_1,\ldots,\nu_L))\\
&\quad \times \prod_{\ell=1}^L\Big(\prod_{\substack{j<\nu_\ell\\ j\neq\nu_1,\ldots,\nu_{\ell-1}}}c(x_j-x_{\nu_\ell})\prod_{\substack{j>\nu_\ell\\ j\neq\nu_{\ell+1},\ldots,\nu_L}}c(x_{\nu_\ell}-x_j)\Big).
\end{split}
\ee
Multiplying \eqref{cLNExpr5} by the prefactors in \eqref{LHSExp} and using \eqref{ratExpr}--\eqref{CNfact}, \eqref{u} and \eqref{rENhom}, we arrive at the right-hand side of \eqref{rENrep2}.
\end{proof}

We proceed to analyze the asymptotic behavior of $\rE_N(x,y)$ for $d_N(y)\to\infty$ using the representation \eqref{rENrep2}.
To this end we need several bounds on the $c$- and $u$-functions, which we derive from the asymptotic estimates~\eqref{casym} and~\eqref{uasym}. 

First,  combining \eqref{casym} with holomorphy of $c(b;z)$ for $(b,\im z)\in S_a\times (0,a_s)$, we obtain a majorization
\be\label{GratBou}
|c(b;p+ir)|\leq c(r,b)\exp(-\gamma |p|),\ \ \ (r,b,p)\in (0,a_s)\times S_a\times\R,
\ee
where we have set
\be
\ga\equiv \alpha\re b/2=\frac{\pi\re b}{a_+a_-},
\ee
and where  $c(r,b)$ is continuous on $(0,a_s)\times S_a $. Likewise, recalling $G(ia)=0$, we get
\be\label{cBou}
|1/c(b;z)|\leq C(b)|\sinh(\ga z)|,\ \ \ (b,z)\in S_a\times\R,
\ee
with $C(b)$ continuous on $S_a $. Finally, letting $b\in S_a$, we note that $1/c(b;z)$ is holomorphic for $\im z\in (-2a,\re b)$. Combining this with~\eqref{casym}, we conclude
\be\label{cbo}
|1/c(b;z)|\le c(b)\exp (\gamma |\re z|),\ \ \ (b,\im z)\in S_a\times [-a_s,0],
\ee
with $c(b)$ continuous on $S_a$.

Turning to the $u$-function~\eqref{u}, we let  $b\in S_a$. Then $u(b;z)$ is holomorphic in the strip $\im z\in (-\min(\re b,2a-\re b),a_s)$. Combining this with~\eqref{uasym}, we readily infer 
 \be\label{ubo}
 |u(b;-z)|\le c(b,\im z),\ \ \ \ (b,\im z)\in S_a\times (-a_s,0],
 \ee
where $c(b,\im z)$ is continuous on $S_a\times (-a_s,0]$.

With these preliminaries out of the way, we return to the function~$\rE_N(x,y)$.
Recalling the symmetry relation $\phi(2a-b)=\phi(b)$ (cf.~\eqref{phi}) and combining this with~\eqref{casym} and~\eqref{cBou}, we find 
\be\label{MNas}
|M_N(b;y)-1|\le  c(b, \rho)\exp(-\alpha\rho d_N(y)),\ \ \ (b, y,\rho)\in S_a\times \R^N\times  [a_s/2,a_s),\ \ \  d_N(y)\ge 0,
\ee
where $c(b, \rho)$ is continuous on $S_a\times [a_s/2,a_s)$.
Moreover, by the induction assumption, we may invoke Theorem \ref{Thm:ubound} after substituting $N\to N-1$. Combining the resulting bound on $\rE_{N-1}$ with the $c$-function estimates just assembled, it is readily verified that both $\rho_N(y)^{-1}\rI_N(x,y,z)$ and $\rho_N(y)^{-1}\hat{\rI}_{N;\nu_1,\ldots,\nu_L}(x,y,z_{>L})$, $L=1,\ldots,N-2$, decay exponentially as $d_N(y)\to\infty$.
 This suggests that the dominant asymptotics of $\rE_N(x,y)$
arises from the last sum in \eqref{rENrep2}.
 
To show that this is indeed the case, we first  observe that the function $\rE_{N-1}^{{\rm as}}(z,w)$ \eqref{ENsc} can be rewritten
\be\label{Easym}
\rE_{N-1}^{{\rm as}}(z,w)=\sum_{\tau\in S_{N-1}}\frac{C_{N-1}(z_\tau)}{C_{N-1}(z)}\exp(i\alpha z_\tau\cdot w).
\ee
Next, taking $N\to N-1$ in Theorem \ref{Thm:rENas}, we deduce from the induction assumption and~\eqref{Easym} that we have
\begin{multline}
\exp(i\alpha y_N(x_1+\cdots+x_N))\rE_{N-1}(x(\nu),(y_1-y_N,\ldots,y_{N-1}-y_N))\\
=\sum_{\substack{\sigma\in S_N\\ \sigma(N)=\nu}}\frac{C_{N-1}(x_{\sigma(1)},\ldots,x_{\sigma(N-1)})}{C_{N-1}(x(\nu))}\exp(i\alpha x_\sigma\cdot y)+R_\nu(x,y),
\end{multline}
where the remainder satisfies a bound
\be\label{RnuBou}
|R_\nu(b;x,y)|\leq C(r,b)P_{N-1}(\gamma |x(\nu)_1|,\ldots,\gamma |x(\nu)_{N-1}|) \exp(-\alpha rd_{N-1}(y_1,\ldots,y_{N-1})),
\ee
which holds for all $(b,x,y)\in S_l\times\R^N\times\R^N$ with $d_{N-1}(y_1,\ldots,y_{N-1})\geq 0$. Here $C(r,b)$ is continuous on $ [a_s/2,a_s)\times S_l$ and $P_{N-1}$ is a polynomial of degree $\leq (N-1)(N-2)/2$ with positive and constant coefficients. Now, for any $\sigma\in S_N$ such that $\sigma(N)=\nu$, we have an identity
\begin{eqnarray}
\frac{C_N(x(\nu),x_\nu)C_{N-1}(x_{\sigma(1)},\ldots,x_{\sigma(N-1)})}{C_{N-1}(x(\nu))} & = &\prod_{\substack{j=1\\ j\neq\nu}}^N c(x_j-x_{\sigma(N)})\cdot \prod_{1\leq j<k\leq N-1}c(x_{\sigma(j)}-x_{\sigma(k)})
\nonumber \\
&  =  &  C_N(x_\sigma).
\end{eqnarray}
Thus we obtain, using~\eqref{Easym} with $N-1\to N$,
\begin{multline}\label{rEN-1sum}
\exp(i\alpha y_N(x_1+\cdots+x_N))\sum_{\nu=1}^N\frac{C_N(x(\nu),x_\nu)}{C_N(x)}\rE_{N-1}(x(\nu),(y_1-y_N,\ldots,y_{N-1}-y_N))\\
=\sum_{\sigma\in S_N}\frac{C_N(x_\sigma)}{C_N(x)}\exp(i\alpha x_\sigma\cdot y)+R(x,y)
=E_N^{{\rm as}}(x,y)+R(x,y),
\end{multline}
with remainder
\be
R(x,y)\equiv \sum_{\nu=1}^N \frac{C_N(x(\nu),x_\nu)}{C_N(x)}R_\nu(x,y).
\ee
We note that an exponential decay bound for $R$ is readily inferred from the bound \eqref{RnuBou} for $R_\nu$. Indeed, after multiplying $|R_{\nu}|$ by $|C_N(x(\nu),x_\nu)/C_N(x)|$ and summing over $\nu=1,\ldots,N$, we need only invoke the $u$-bound \eqref{ubo}.

In the following theorem our starting point is \eqref{rENrep2}, rewritten as
\begin{multline}\label{rENrep3}
(\rE_N-\rE_N^{\rm as})(x,y)=(M_N(y)-1)\rE_N^{\rm as}(x,y)+M_N(y)R(x,y) \\
+\exp(i\alpha y_N(x_1+\cdots+x_N))\frac{M_N(y)}{\rho_N(y)}\Bigg[\frac{1}{(N-1)!}\left(\frac{G(ib-ia)}{\sqrt{a_+a_-}}\right)^{N-1}\int_{(C_b+ir)^{N-1}}dz\, \rI_N(x,y,z)\\
+\sum_{L=1}^{N-2}\frac{1}{(N-1-L)!}\left(\frac{G(ib-ia)}{\sqrt{a_+a_-}}\right)^{N-1-L}\sum_{1\leq\nu_1<\cdots<\nu_L\leq N}\cU_{\nu_1,\ldots,\nu_L}(x)\\
\times \int_{(C_b+ir)^{N-1-L}}dz_{>L}\, \hat{\rI}_{N;\nu_1,\ldots,\nu_L}(x,y,z_{>L})\Bigg] ,
\end{multline}
where we have used~\eqref{rEN-1sum}. In view of our considerations above, we need only majorize the expression in square brackets on the right-hand side to infer exponential decay of the left-hand side with rate $\alpha r$ as $d_N(y)\to\infty$. As an immediate corollary, we obtain the `unitary asymptotics' \eqref{ENsc} of $\rE_N$.
 
\begin{theorem}\label{Thm:rENas}
Letting $(r,b)\in [a_s/2,a_s)\times S_l$, we have
\be
|(\rE_N-\rE_N^{\rm as})(b;x,y)|<C(r,b)P_N(\gamma |x_1|,\ldots,\gamma |x_N|) \exp(-\alpha rd_N(y)),
\ee
for all $x,y\in\R^N$ with $d_N(y)>0$, where $C$ is continuous on $[a_s/2,a_s)\times S_l$ and $P_N$ is a polynomial of degree $\leq N(N-1)/2$ with positive and constant coefficients.
\end{theorem}

\begin{proof}
In view of~\eqref{ubo} (with $\im z =0$),  it suffices to establish the bounds
\be\label{INbou}
\left|\int_{(C_b+ir)^{N-1}}dz\, \rI_N(x,y,z)\right|\leq C_0(r,b)|\rho_N(y)|P_{N,0}(\gamma |x_1|,\ldots,\gamma |x_N|)\exp(-\alpha rd_N(y)),
\ee
\begin{multline}\label{hrINbou}
\left|\int_{(C_b+ir)^{N-1-L}}dz_{>L}\, \hat{\rI}_{N;\nu_1,\ldots,\nu_L}(x,y,z_{>L})\right|\\
\leq C_L(r,b)|\rho_N(y)|P_{N,L}(\gamma |x_1|,\ldots,\gamma |x_N|)\exp(-\alpha rd_N(y)),\ \ \ L=1,\ldots,N-2,
\end{multline}
for all $x,y\in\R^N$ with $d_N(y)>0$. Here the functions $C_0$, $C_L$ are continuous on $[a_s/2,a_s)\times S_l$ and $P_{N,0}$, $P_{N,L}$ are polynomials of degree $\leq N(N-1)/2-L$ with positive and constant coefficients.

Taking $z_k\to z_k+i(a-b/2+r)$, we infer from the identity \eqref{rENhom} with $N\to N-1$ that
\begin{multline}\label{rINexpr}
\int_{(C_b+ir)^{N-1}}dz\, \rI_N(x,y,z)=\rho_N(y)\exp\Big(-\alpha r\sum_{m=1}^{N-1}(y_m-y_N)\Big)C_N(x)^{-1}\\
\times \int_{\R^{N-1}}dz\, \frac{\rE_{N-1}(z,(y_1-y_N,\ldots,y_{N-1}-y_N))}{C_{N-1}(-z)}\prod_{j=1}^N\prod_{k=1}^{N-1}c(z_k+ir-x_j).
\end{multline}
Now by the induction assumption, Theorem \ref{Thm:ubound} holds true when $N$ is replaced by $N-1$. Combining the resulting bound on $\rE_{N-1}$ with \eqref{GratBou} and~\eqref{cBou}, we deduce
\be
\begin{split}
\left|\int_{(C_b+ir)^{N-1}}dz\, \rI_N(x,y,z)\right| &\leq C_0(r,b)|\rho_N(y)|\exp\Big(-\alpha r\sum_{m=1}^{N-1}(y_m-y_N)\Big)\\
&\quad \times \int_{\R^{N-1}}dz\, P_{N-1}(\gamma |z_1|,\ldots,\gamma |z_{N-1}|)\exp( F_{N-1}(\gamma x,\gamma z)),
\end{split}
\ee
where $F_{N-1}$ is given by \eqref{FL} and $P_{N-1}$ is a polynomial of degree $\leq (N-1)(N-2)/2$ with positive and constant coefficients. The bound \eqref{INbou} is now a direct consequence of Theorem \ref{Thm:IPLbou}.

We proceed to prove \eqref{hrINbou}. Taking $z_k\to z_k+i(a-b/2+r)$, $L<k\leq N-1$, and using once more \eqref{rENhom}, we obtain
\begin{multline}\label{rhINexpr}
\int_{(C_b+ir)^{N-1-L}}dz_{>L}\, \hat{\rI}_{N;\nu_1,\ldots,\nu_L}(x,y,z_{>L})=\rho_N(y)C_{N-L}(x(\nu_1,\ldots,\nu_L))^{-1}\\
\times \int_{\R^{N-1-L}}dz_{>L}\, \rE_{N-1}((x_{\nu_1},\ldots,x_{\nu_L},z_{L+1}+ir,\ldots,z_{N-1}+ir),(y_1-y_N,\ldots,y_{N-1}-y_N))\\
\times \frac{1}{C_{N-1-L}(-z_{>L})}
 \prod_{\substack{j=1\\ j\neq \nu_1,\ldots,\nu_L}}^N\prod_{k=L+1}^{N-1}c(z_k+ir-x_j).
\end{multline}
By Theorem \ref{Thm:ubound} with $N\to N-1$ and \eqref{GratBou}--\eqref{cBou}, it follows that
\begin{multline}\label{hrINbou2}
\left|\int_{(C_b+ir)^{N-1-L}}dz_{>L}\, \hat{\rI}_{N;\nu_1,\ldots,\nu_L}(x,y,z_{>L})\right|\leq C_L(r,b)|\rho_N(y)|\exp\Big(-\alpha r\sum_{m=L+1}^{N-1}(y_m-y_N)\Big)\\
\times \int_{\R^{N-1-L}}dz_{>L}\, P_{N-1}(\gamma |x_{\nu_1}|,\ldots,\gamma |x_{\nu_L}|,\gamma |z_{L+1}|, \ldots,\gamma |z_{N-1}|)\\
\times \exp\big( F_{N-1-L}(\gamma x(\nu_1,\ldots,\nu_L),\gamma z_{>L})\big).
\end{multline}
Since $P_{N-1}$ is a polynomial of degree $\leq (N-1)(N-2)/2$ with positive, constant coefficients, we have
\begin{multline}\label{PN-1exp}
P_{N-1}(\gamma |x_{\nu_1}|,\ldots,\gamma |x_{\nu_L}|,\gamma |z_{L+1}|,\ldots,\gamma |z_{N-1}|)\\
=\sum_{\substack{k\in\N^L\\ |k|\leq (N-1)(N-2)/2}} \gamma^{|k|}|x_{\nu_1}|^{k_1}\cdots |x_{\nu_L}|^{k_L} P^k_{N-1,L}(\gamma |z_{L+1}|,\ldots,\gamma |z_{N-1}|),
\end{multline}
for some polynomials $P^k_{N-1,L}$ of degree $\leq (N-1)(N-2)/2-|k|$ with positive, constant coefficients, where $|k|\equiv k_1+\cdots+k_L$. Substituting this expansion in \eqref{hrINbou2}, we can use Theorem \ref{Thm:IPLbou} to bound each term separately. Indeed, from \eqref{IPL}--\eqref{IPLbo} we get
\begin{multline}\label{intPbou}
\int_{\R^{N-1-L}}dz_{>L}\, P^k_{N-1,L}(\gamma |z_{L+1}|,\ldots,\gamma |z_{N-1}|)\exp\big(F_{N-1-L}(\gamma x(\nu_1,\ldots,\nu_L),\gamma z_{>L})\big)\\
< P^k_{N,L}((\gamma |x_j|)_{j\neq \nu_1,\ldots,\nu_L}),
\end{multline}
for some polynomials $P^k_{N,L}$ of degree
\be
\deg\, P^k_{N,L}\leq (N-1)(N-2)/2-|k|+N-1-L=N(N-1)/2-|k|-L,
\ee
with positive, constant coefficients. The bounds \eqref{hrINbou2} and \eqref{intPbou} clearly imply the desired majorization \eqref{hrINbou}.
\end{proof}

We proceed to obtain a bound on $\rE_N(x,y)$ for $x,y\in\C^N\times\R^N$ satisfying
\be\label{xyres}
v_j-v_k\in(-a_s,0],\ \ 1\leq j<k\leq N,\ \ d_N(y)>0, \ \ \ v=\im x.
\ee
Like in the $N=2$ and $N=3$ cases treated in II, we take as a starting point the representation for $\rE_N$ given by \eqref{rENrep2}.   

We first derive the desired bound for the last sum in~\eqref{rENrep2}.  To begin with, from \eqref{MNas}  we easily get
\begin{eqnarray}\label{MNexp}
|M_N(b;y)\exp(i\alpha y_N(x_1+\cdots +x_N))|& < &  c(b)\exp\Big(-\alpha\sum_{j=1}^Ny_jv_j\Big)
\nonumber \\
&  \times & \exp\Big(\alpha\sum_{k=1}^{N-1}(y_k-y_N)v_k\Big),
\end{eqnarray}
for all $(b,x,y)\in S_a\times \C^N\times\R^N$, with $c(b)$ continuous on $S_a$. Using next Theorem~\ref{Thm:ubound} with $N\to N-1$, we get an estimate 
\begin{eqnarray}\label{ENbo}
|\rE_{N-1}(x(\nu),(y_1-y_N,\ldots,y_{N-1}-y_N))| & < &  C(\de,b)P_{N-1}(\gamma |\re x(\nu)_1|,\ldots,\gamma |\re x(\nu)_{N-1}|)
\nonumber \\
 & \times & \exp\Big(-\alpha\sum_{k=1}^{N-1}(y_k-y_N)\im x(\nu)_k\Big),
\end{eqnarray}
where $P_{N-1}$ is a polynomial of degree $\leq (N-1)(N-2)/2$ with positive and constant coefficients.
Now when we take the product $\Pi_{\nu}$ of the functions on the left-hand sides of~\eqref{MNexp} and~\eqref{ENbo}, we can use the majorization
\begin{multline}
\exp\Big(\alpha\sum_{k=1}^{N-1}(y_k-y_N)v_k\Big) \exp\Big(-\alpha\sum_{k=1}^{N-1}(y_k-y_N)\im x(\nu)_k\Big)\\
=\exp\Big(\alpha\sum_{k=\nu }^{N-1}(y_k-y_N)(v_k-v_{k+1})\Big)\le  1, \ \ \ d_N(y)>0, \ v_k-v_{k+1}\le 0  ,\ k=1,\ldots, N-1,
\end{multline}
to conclude that the product of $\Pi_{\nu}$ and the pertinent $u$-function product satisfies a bound of the type occurring in~\eqref{rENubou}, cf.~\eqref{rENrep2} and~\eqref{ubo}. (Indeed,  from \eqref{u} and the $G$-pole locations I (A.11), we infer regularity of $u(b;x_k-x_j)$ for $-a_s<v_j-v_k<\min(\re b,2a-\re b)$.) 

\begin{theorem}\label{Thm:ubound}
Letting $(\de,b)\in (0,a_s]\times S_l$, we have
\be\label{rENubou}
|\rE_N(b;x,y)|<C(\de,b)P_N(\gamma |\re x_1|,\ldots,\gamma |\re x_N|)\exp\Big(-\alpha\sum_{j=1}^N y_jv_j\Big),
\ee
for all $(x,y)\in\C^N\times\R^N$ satisfying
\be\label{xyas}
v_j-v_k\in[-a_s+\de,0] ,\ \ 1\leq j<k\leq N,\ \ \ d_N(y)>0,\ \ \ v=\im x,
\ee
where $C(\de,b)$ is a continuous function on $(0,a_s]\times S_l$ and $P_N$ is a polynomial of degree $\leq N(N-1)/2$ with positive and constant coefficients.
\end{theorem}

\begin{proof}
 Since we have already shown that the last sum in~\eqref{rENrep2}  satisfies a bound of this type, the assertion will follow once we prove that the integrals on the right-hand side of \eqref{rENrep2} are bounded by
\be\label{major}
C(\de,b)|\rho_N(b;y)|P_N(\gamma |\re x_1|,\ldots,\gamma |\re x_N|)\exp\Big(-\alpha\sum_{k=1}^{N-1}(y_k-y_N)v_k\Big),
\ee
for all $(x,y)\in\C^N\times\R^N$ satisfying \eqref{xyas}. Indeed, by the induction assumption, \eqref{rENubou} holds true with $N$ replaced by $N-1$, and when combined with the $c$-bound~\eqref{cbo}, it becomes clear that we can find a polynomial $P_N$ of the required form such that the remaining sum is majorized by \eqref{major} without the factor $|\rho_N(b;y)|$.

Due to the identity \eqref{rENhom}, we may and shall restrict attention to
\be\label{vres}
0\leq v_1\leq\cdots\leq v_N\leq a_s-\de.
\ee
Requiring at first $x\in\R^N$, we repeat the steps leading to the $(N-1)$-fold integral~\eqref{rINexpr}. Allowing next $v_j\neq 0$, we require
\be\label{vres2}
\de^\prime\leq r-v_j\leq a_s-\de^\prime,\ \ \ \de^\prime\in(0,a_s/2],\ \ j=1,\ldots,N,
\ee
so that we stay clear of the poles of the $c$-functions for $z_k+ir-x_j=0,a_s$. Choosing
\be\label{rde}
r=a_s-\de/2,\ \ \ \de^\prime=\de/2,
\ee
we can allow any $x\in\C^N$ satisfying \eqref{vres}. Invoking \eqref{rENubou} with $N\to N-1$ and the bounds \eqref{GratBou}--\eqref{cBou}, we thus infer
\begin{multline}
\left|\int_{(C_b+ir)^{N-1}}dz\, \rI_N(x,y,z)\right|\leq c_2(\de,b)|\rho_N(y)|\exp\Big(-\alpha r\sum_{k=1}^{N-1}(y_k-y_N)\Big)\\
\times \int_{\R^{N-1}}dz\, P_{N-1}(\gamma |z_1|,\ldots,\gamma |z_{N-1}|)\exp\big( F_{N-1}((\gamma \re x_1,\ldots,\gamma \re x_N),\gamma z)\big),
\end{multline}
where $c_2$ is continuous on $(0,a_s]\times S_l$. Using Theorem \ref{Thm:IPLbou} to bound the remaining integral, we arrive at the desired majorization.

We turn now to the $(N-1-L)$-fold integral~\eqref{rhINexpr}. Assuming \eqref{vres2}--\eqref{rde}, we can again allow any $x\in\C^N$ satisfying \eqref{vres}. Indeed, we stay clear of the pertinent poles of the $c$-functions and can use \eqref{rENubou} with $N\to N-1$ and $\de\to\de/2$ to bound the $\rE_{N-1}$-factor. Using also the bounds \eqref{GratBou} and~\eqref{cbo}, we obtain
\begin{multline}\label{hrINbou3}
\left|\int_{(C_b+ir)^{N-1-L}}dz_{>L}\, \hat{\rI}_{N;\nu_1,\ldots,\nu_L}(x,y,z_{>L})\right|<c_3(\de,b)|\rho_N(y)|\\
\times \exp\Big(-\alpha\sum_{j=1}^L(y_j-y_N)v_{\nu_j}-\alpha r\sum_{k=L+1}^{N-1}(y_k-y_N)\Big)\\
\times  \int_{\R^{N-1-L}}dz_{>L}\, P_{N-1}(\gamma |\re x_{\nu_1}|,\ldots,\gamma |\re x_{\nu_L}|,\gamma |z_{L+1}|,\ldots,\gamma |z_{N-1}|)\\
\times \exp\big( F_{N-1-L}(\gamma \re x(\nu_1,\ldots,\nu_L),\gamma z_{>L})\big),
\end{multline}
with $c_3$ continuous on $(0,a_s]\times S_l$. Now we have 
\be
v_{\nu_j}\geq v_j,  \ j=1,\ldots, L,\ \ \ \ r>v_j,\  j=1,\ldots,N, \ \ \  \ d_N(y)>0,
\ee
whence we infer
\be
\exp\Big(-\alpha\sum_{j=1}^L(y_j-y_N)v_{\nu_j}-\alpha r\sum_{k=L+1}^{N-1}(y_k-y_N)\Big)<\exp\Big(-\alpha\sum_{k=1}^{N-1}(y_k-y_N)v_k\Big).
\ee
Also, substituting the expansion \eqref{PN-1exp} with $x_{\nu_j}\to \re x_{\nu_j}$ in \eqref{hrINbou3}, each term is readily bounded using Theorem \ref{Thm:IPLbou}. Hence the majorization \eqref{major} results.
\end{proof}

\begin{appendix}

\section{Polynomial bounds}\label{AppA}
In Section \ref{Sec2} we use the following theorem to bound remainder terms when studying the asymptotic behavior of the functions $\rE_N$, cf.~Theorems \ref{Thm:rENas}--\ref{Thm:ubound}.

\begin{theorem}\label{Thm:IPLbou}
Let $z_1,\ldots,z_L,u_1,\ldots,u_{L+1}\in\R$, and let $\cP_{L,M}(|z_1|,\ldots,|z_L|)$ be a polynomial of degree~$M$ with positive coefficients. Setting
\be\label{IPL}
I_{\cP,L}(u_1,\ldots,u_{L+1})\equiv \int_{\R^L}dz\, \cP_{L,M}(|z_1|,\ldots,|z_L|)\exp( F_L(u,z)),
\ee
where
\be\label{FL}
F_L(u,z)\equiv \sum_{1\le m<n\le L+1}|u_m-u_n|+
\sum_{1\le m<n\le L}|z_m-z_n|-\sum_{j=1}^{L+1}\sum_{k=1}^{L}|u_j-z_k|,
\ee
we have a bound
\be\label{IPLbo}
I_{\cP,L}(u_1,\ldots,u_{L+1})<Q_{L,M}(|u_1|,\ldots,|u_{L+1}|),
\ee
where $Q_{L,M}$ is a polynomial of degree $\le M+L$ with positive coefficients.
\end{theorem}
\begin{proof}
We prove this by induction on~$L$. For $L=1$ we have
\be
I_{\cP,1}(u_1,u_2)=\int_{\R}dz\, \cP_{1,M}(|z|)\exp (|u_1-u_2|-|u_1-z|-|u_2-z|).
\ee
We have symmetry under swapping $u_1$ and~$u_2$, so we may take $u_2\le u_1$. We write the integral as the sum of three integrals over $(-\infty, u_2)$, $[u_2,u_1]$ and $(u_1,\infty)$, denoted by $I^-$, $I^{\mu}$ and $I^+$, resp. Then we have
\be
I^+=\int_{u_1}^{\infty}dz\, \cP_{1,M}(|z|)\exp(u_1-u_2-(z-u_1)-(z-u_2))=\int_0^{\infty}dz\, \cP_{1,M}(|z+u_1|)e^{-2z}.
\ee
Now we need only use $|z+u_1|\le z+|u_1|$ to see that $I^+$ is bounded by a polynomial of degree~$M$ in $|u_1|$ with positive coefficients.

Likewise, since
\be
I^-=\int_{-\infty}^{u_2}dz\, \cP_{1,M}(|z|)\exp(u_1-u_2-(u_1-z)-(u_2-z))=\int_{-\infty}^0dz\, \cP_{1,M}(|z+u_2|)e^{2z},
\ee
we infer that $I^-$ is bounded by a polynomial of degree~$M$ in $|u_2|$ with positive coefficients.

Finally, we have for the middle integral
\be
I^{\mu}=\int_{u_2}^{u_1}dz\, \cP_{1,M}(|z|)\exp(u_1-u_2-(u_1-z)-(z-u_2))= \int_{u_2}^{u_1}dz\, \cP_{1,M}(|z|),
\ee
and since we have
\be
\int_{u_2}^{u_1}dz\, |z|^k\le \frac{1}{k+1}\Big(|u_1|^{k+1}+|u_2|^{k+1}\Big),\ \ k\in\N,
\ee
we see that $I^{\mu}$ is bounded by a polynomial of degree~$M+1$ in $|u_1|, |u_2|$, with positive coefficients. Thus the assertion holds true for $L=1$.

Next, we inductively assume the assertion has been proved up to $L-1$, $L>1$. First, we claim that the function~$F_L(u,z)$~\eqref{FL} 
satisfies
\be\label{FLbo}
F_L(u,z)\le 0,\ \ \ \forall (u,z)\in\R^{L+1}\times\R^L.
\ee
Clearly, $F$ has permutation symmetry in $u_1,\ldots,u_{L+1}$ and in $z_1,\ldots,z_L$.  Therefore, we need only prove~\eqref{FLbo} under the assumptions  $z_L\le z_{L-1}\le \cdots \le z_1$ and 
\be\label{uas}
u_{L+1}\le u_L\le \cdots \le u_1. 
\ee
Then we have
\be
F_L(u,z)\le \sum_{1\le m<n\le L+1}(u_m-u_n)+
\sum_{1\le m<n\le L}(z_m-z_n)-\sum_{j=1}^{L+1}\Big(\sum_{j\le k}(u_j-z_k)+\sum_{j>k}(z_k-u_j)\Big)=0,
\ee
and so~\eqref{FLbo} follows.

We are now prepared to prove the bound~\eqref{IPLbo}. By permutation invariance of~$I_{\cP,L}(u)$, we need only show its validity under the assumption~\eqref{uas}. We write each $z_k$-integral as the sum of three integrals over $(-\infty, u_{L+1})$, $[u_{L+1},u_1]$ and $(u_1,\infty)$, denoted by $I^-$, $I^{\mu}$ and $I^+$, resp. We denote by~$\hat{z}^k$ the vector in~$\R^{L-1}$ arising by omitting the coordinate $z_k$ from $z\in\R^L$. Then we have
\begin{multline}
I_{\cP,L}(u)=\Big(\prod_{k=1}^N\big( I^-+I^{\mu}+I^+\big)dz_k \Big) \cP\exp(F_L)\\
< \sum_{k=1}^L\Big( I^- dz_k\int_{\R^{L-1}}d\hat{z}^k +I^+ dz_k\int_{\R^{L-1}}d\hat{z}^k\Big) \cP\exp(F_L)\\
+\prod_{k=1}^LI^{\mu}dz_k\, \cP\exp(F_L).
\end{multline}

Next, using the bound~\eqref{FLbo}, we note that the integral over $[u_{L+1},u_1]^L$ is bounded by a sum of terms of the form
\be
c\prod_{k=1}^LI^{\mu}dz_k\, |z_k|^{n_k},\ \ c>0,\ \ \sum_{k=1}^Ln_k\le M.
\ee
In turn, such a term is bounded by
\be
c_n\prod_{k=1}^L\big(|u_1|^{n_k+1}+|u_{L+1}|^{n_k+1}\big),\ \ c_n>0.
\ee
Hence the integral over $[u_{L+1},u_1]^L$ is majorized by a polynomial in~$|u_1|,|u_{L+1}|$ of degree $\le M+L$ with positive coefficients. 

We proceed to study the $z_k$-integral $I^+$.
We have $u_1<z_k$, so we may write $F_L$ as
\be
\sum_{j=2}^{L+1}(u_1-u_j)+\sum_{l\ne k}|z_k-z_l|-\sum_{j=1}^{L+1}(z_k-u_j)-\sum_{l\ne k} |u_1-z_l|+F_{L-1}^+((u_2,\ldots,u_{L+1}),\hat{z}^k).
\ee
Taking $z_k\to z_k+u_1$ in the integral, we then get
\be
e^{F_{L-1}^+}\int_0^{\infty}dz_k \, \cP(|z_1|,\ldots, |z_k+u_1|,\ldots, |z_L|)\exp\Big(-(L+1)z_k
+\sum_{l\ne k}(|z_k+u_1-z_l|-|u_1-z_l|)\Big).
\ee
Majorizing the exponential by $\exp(-2z_k)$, we can bound each monomial term as a polynomial in $|u_1|$ of degree $\le M$. The induction assumption now applies to the remaining $\hat{z}^k$-integrals over $\R^{L-1}$, yielding polynomials of the announced form. 

The $L$ integrals  $I^-dz_k$ can be estimated in a similar way, first writing $F_L$ as 
\be
\sum_{j=1}^{L}(u_j-u_{L+1})+\sum_{l\ne k}|z_k-z_l|-\sum_{j=1}^{L+1}(u_j-z_k)-\sum_{l\ne k} |u_{L+1}-z_l|+F_{L-1}^-((u_1,\ldots,u_{L}),\hat{z}^k),
\ee
and then taking $z_k\to z_k+u_{L+1}$. 
\end{proof}

\end{appendix}

\bibliographystyle{amsalpha}

\end{document}